\documentclass[sigconf]{aamas} 

\setcopyright{ifaamas}
\acmConference[AAMAS '24]{Under review for the 23rd International Conference
on Autonomous Agents and Multiagent Systems (AAMAS 2024)}{May 6 -- 10, 2024}
{Auckland, New Zealand}{N.~Alechina, V.~Dignum, M.~Dastani, J.S.~Sichman (eds.)}
\copyrightyear{2024}
\acmYear{2024}
\acmDOI{}
\acmPrice{}
\acmISBN{}

\acmSubmissionID{851}

\usepackage{balance}
  \usepackage{todonotes}
\usepackage{caption}
\usepackage{subcaption}

\usepackage{comment}
\usepackage{amsthm}
\usepackage{mathtools}
\usepackage{acronym}
 \usepackage{mathrsfs}  
\usepackage{bm}
\usepackage{amsmath,amsfonts}
\usepackage{tikz}
\usetikzlibrary{automata,positioning}
\usetikzlibrary{arrows}
\usepackage{float}
\newif\ifuseboldmathops
\newif\ifuseittextabbrevs
\useboldmathopstrue   

\ifuseittextabbrevs

\else

\fi

\ifuseboldmathops
	\newcommand{\reals}{\mathbf{R}}

\else
	\newcommand{\reals}{\mathbb{R}}

\fi

\ifuseboldmathops

\else

\fi

\ifuseboldmathops
	\newcommand{\Expect}{\mathop{\bf E{}}\nolimits}
	
\else
	\newcommand{\Expect}{\mathop{\mathbb{E}{}}\nolimits}
	
\fi

\DeclareMathOperator*{\optmax}{\mathrm{maximize}}

\DeclareMathOperator*{\optsts}{\textrm{s.t.}}
\ifuseboldmathops

\else

\fi

\newcommand{\obs}{\mathsf{Obs}}

\newcommand{\emission}{\mathbf{E}}
\newcommand{\trans}{\mathbf{P}}

 \acrodef{mdp}[MDP]{Markov Decision Process}
  \acrodef{hmm}[HMM]{Hidden Markov Model}

\acrodef{pomdp}[POMDP]{Partially Observable MDP}
 \acrodef{tl}[TL]{Temporal Logic}
  \acrodef{vm}[VM]{Virtual Machine}
  \acrodef{mtd}[MTD]{Moving Target Defense}
    \acrodef{pctl}[PCTL]{Probabilistic Computation Tree Logic}
    
        \acrodef{sdn}[SDN]{software-defined networking}

\theoremstyle{definition}
 \newtheorem{definition}{Definition}

\newtheorem{remark}{Remark}
\newtheorem{theorem}{Theorem}

\newcommand{\dist}[1]{\mathcal{D}(#1)}

\usepackage{graphicx}
\usepackage{algpseudocode}
\usepackage{algorithm}
\usepackage{adjustbox}

 \definecolor{darkgreen}{rgb}{0,0.5,0}
  \definecolor{darkblue}{rgb}{0.36,0.54,0.66}

\title{Covert Planning against Imperfect Observers}

\author{Haoxiang Ma}
\affiliation{
  \institution{University of Florida}
  \city{Gainesville, FL}
  \country{United State}}
\email{hma2@ufl.edu}

\author{Chongyang Shi}
\affiliation{
  \institution{University of Florida}
  \city{Gainesville, FL}
  \country{United State}}
\email{c.shi@ufl.edu}

\author{Shuo Han}
\affiliation{
  \institution{University of Illinois Chicago}
  \city{Chicago, IL}
  \country{United State}}
\email{hanshuo@uic.edu}

\author{Michael R. Dorothy}
\affiliation{
  \institution{DEVCOM Army Research Laboratory}
  \city{Adelphi, MD}
  \country{United State}}
\email{michael.r.dorothy.civ@army.mil}

\author{Jie Fu}
\affiliation{
  \institution{University of Florida}
  \city{Gainesville, FL}
  \country{United State}}
\email{fujie@ufl.edu}

\begin{abstract}

Covert planning refers to a class of constrained planning problems where an agent aims to accomplish a task with minimal information leaked to a passive observer to avoid detection. However, existing methods of covert planning often consider deterministic environments or do not exploit the observer's imperfect information.  This paper studies how covert planning can leverage the coupling of stochastic dynamics and the observer's imperfect observation to achieve optimal task performance without being detected. Specifically, we employ a Markov decision process to model the interaction between the agent and its stochastic environment, and a partial observation function to capture the leaked information to a passive observer. Assuming the observer employs hypothesis testing to detect if the observation deviates from a nominal policy, the covert planning agent aims to maximize the total discounted reward while keeping the probability of being detected as an adversary below a given threshold. We prove that finite-memory policies are more powerful than Markovian policies in covert planning. Then, we develop a primal-dual proximal policy gradient method with a two-time-scale update to compute a (locally) optimal covert policy. We demonstrate the effectiveness of our methods using a stochastic gridworld example. Our experimental results illustrate that the proposed method computes a policy that maximizes the adversary's expected reward without violating the detection constraint,   and empirically demonstrates how the environmental noises can influence the performance of the covert policies.

\end{abstract}

\keywords{Markov decision processes, Deception, Covert Planning}

\begin{document}
\maketitle

\section{Introduction}  
Covert planning refers to accomplishing some tasks with minimal information leaked to a passive observer to avoid detection. Such planning algorithms are useful in various security applications including surveillance and crime prevention. For example, a security patrolling agent (robot or human) would need to minimize the knowledge of his presence while collecting information in several regions of interest. In a contested search and rescue mission, a human-robot team may be sent to gather information in the hostile environment without being detected. Game AI is another area that uses covert planning algorithms \cite{aleneziStealthyPathPlanning2022}. An AI player can use covert planning to cause the element of surprise to the human player, thus making the game more entertaining and realistic.".

In this work, we study a class of covert planning in stochastic environments. Consider a planning problem in a stochastic environment modeled as an \ac{mdp}. The goal of the planning agent is to maximize a  discounted total reward, while ensuring covert behavior against an observer. Specifically, the covert constraint requires that with a high probability, an observer could not detect any deviation from a nominal behavior modeled by a Markov chain, with its imperfect observation of the agent's path.  We employ a sequential likelihood ratio test to construct the covert constraint and prove that a finite-memory policy can be more powerful than a Markovian policy in the resulting constrained \ac{mdp}. Due to the intractable search space for finite-memory policies, we develop a primal-dual gradient-based policy search method to compute an optimal and covert Markov policy. To mitigate the distribution draft and improve the stability, we employ a two-time-scale approach and a sample-efficient estimation for the policy gradient. We demonstrate the performance of the proposed algorithms using a security patrolling agent tasked with visiting a set of goal states, while the nominal behavior is obtained from other users that perform routine activities in the dynamic environment. 

\paragraph*{Related Work}
The work on covert path planning is closely related to stealthy evader strategies in the pursuit-evasion game and covert robots \cite{al2011robotic}. 
In a pursuit-evader game, a team of pursuers aims to locate and capture one or more evaders. The pursuit-evasion game can be of imperfect information where the evader hides at or moves between locations unobservable to to pursuer. The interactions are often formulated as a hide-and-seek game \cite{wangLearningHideandseek2014,GabiGoobar1635600}. The equilibrium of pursuit-evasion games with sensing limitation has been investigated for continuous-state dynamical systems \cite{bopardikar2008discrete,gerkey2006visibility} and deterministic games on graphs \cite{islerRandomizedPursuitEvasionLocal2006}.  Our formulation can be viewed as planning for a stealthy evader who aims to achieve the goal while remaining hidden in a stochastic environment. Specifically, the nominal behavior in our setting is simply the environment dynamics without the presence of an evader. However, due to the stochastic dynamics and noisy observation, the covert planning cannot be solved with existing algorithms for hide-and-seek games or plan obfuscation that consider deterministic dynamics. 

For stochastic systems modeled as \ac{mdp}s, deceptive planning has been studied. In \cite{karabagDeceptionSupervisoryControl2021}, the authors study a   deceptive planning problem where a supervisor determines a reference policy for the agent to follow, while the agent instead uses a different, deceptive policy
to achieve a secret task. The planning problem is formulated such that the agent minimizes the divergence between the distributions under the deceptive policy and that under the supervisor's policy, while ensuring the probability of achieving a secret task is greater than a given threshold. The KL-divergence between a policy and 
a reference policy is also used as a metric for deceptiveness in \cite{savasEntropyMaximizationPartially2022,savas2022deceptive} to deceive the supervisor/observer into believing the agent's policy follows a given reference. 
Other related work includes plan recognition \cite{ramirezGoalRecognitionPOMDPs,sukthankar2014plan} where the observer aims to infer the goal of an agent, given a finite set of possible goals and its observation.
 Covert planning can be viewed as counter-plan recognition. 

In comparison to deceptive planning in \ac{mdp}s \cite{karabagDeceptionSupervisoryControl2021,savasEntropyMaximizationPartially2022,savas2022deceptive}, our work differs in the following aspects: 1)  The existing work assumes full observations of the passive observer, whiles we consider observers with imperfect observations and address how the agent can leverage both the noise in the environmental dynamics and imperfect information of the observer for covertness. This explicit consideration of observation functions provides more insight into verifying the effectiveness of a sensor design against deceptive, covert planning adversaries. 2) We employ the likelihood ratio test for anomaly detection to generate the covert constraint, instead of minimizing the policy difference. This approach allows the planning algorithm to explicitly bind the probability of detection.



This paper is organized as follows. We provide the preliminary definitions and formal problem statement in Section~\ref{sec: preliminary}. In Section~\ref{sec: main result}, we analyze the $(1 - \alpha)$-covert planning problem, propose a two-time-scale primal-dual proximal policy gradient method, and compute a (locally) optimal covert policy. In Section ~\ref{sec: experiment}, we demonstrate our results using a stochastic gridworld example. We show our framework maximizes the adversary's expected reward without violating the detection constraints. We conclude in Section ~\ref{sec: conclusion} and discuss future directions.

\section{Preliminary and Problem Formulation}
\label{sec: preliminary}
The planning problem is modeled as a Markov decision process  $M=(S, A, P,s_0, R, \gamma)$ where $S$ is a finite set of states, $A$ is a finite set of actions, $P:S \times A\rightarrow \dist{S}$ is a probabilistic transition function and  $P(s'|s, a)$ is the probability of reaching state $s'$ given that action $a$  is taken at the state $s$. The initial state is $s_0$. 
The planning objective for the agent (referred to as player 1/P1) is described by a reward function $R: S\times A\rightarrow \reals$.

  A policy for $M$ is a function $\pi: \mathscr{D} \rightarrow \mathscr{C}$ where it is called \emph{memoryless or Markovian} if $\mathscr{D}=S$; \emph{finite-memory} if $\mathscr{D}=(S\times A)^* S$; \emph{deterministic} if $\mathscr{C}=A$, and \emph{randomized} if $\mathscr{C}=\dist{A}$.

For a Markovian policy $\pi:S\rightarrow \dist{A}$,  P1's value function  $V  ^{\pi}: S \rightarrow \reals$ is defined as
\[
V ^{\pi}(s) = E_{\pi}[\sum\limits_{k = 0}^{\infty}\gamma^{k}R(s_k, \pi(s_k))|s_0 =s],
\] where $E_{\pi}$ is the expectation with respect to the probability distribution induced by the policy $\pi$ from the \ac{mdp} $M$, and $s_k$ is the $k$-th state in the Markov chain induced from the \ac{mdp} $M$ under the policy $\pi$, starting from state $s$. 

Given the \ac{mdp} $M$, the goal of P1 is to maximize the total discounted reward given some discounting factor $\gamma$. In addition, the agent must ensure its behavior is covert with respect to a passive observer  (player 2/P2) whose imperfect observation function is given as follows. 

\begin{definition}[Observation function of P2]
Let $O$ be a finite set of observations.  The state-observation function of P2 is $\obs_S: S  \rightarrow \dist{O}$ that maps a state $s$ to a distribution $\obs_S(s)$ over observations. The action observation function is state-dependent, defined as $\obs_A: S\times A \rightarrow \dist{O}$ that maps an action $a$ from $s$ to a distribution $\obs_A(s, a)$ over observations. 
\end{definition}
Considering the actions that are state-dependent allows for the most general class of observation functions. 
Without loss of generality, we denote $\obs: S\cup S\times A\rightarrow \dist{O}$ as the combined state and action observation function.

The goal of P2 is to detect if there is any deviation of the agent's behavior from a normal user's behavior. A normal user follows a Markovian policy $\pi_0$ in $M$, referred to as the \emph{nominal policy}. As a result, 
the normal user's behavior is modeled as a hidden Markov model induced from the original \ac{mdp} $M$ given some nominal user policy and the defender's observation function.

\begin{definition}[HMM modeling the P2's observation given the nominal policy $\pi_0$]\label{def:hmm0}

 Given the \ac{mdp} $M =(S,A, P, s_0)$, the nominal policy $\pi_0: S\rightarrow \dist{A}$, and an observation function $\obs:S \times  A\rightarrow \dist{O}$, the stochastic process of observations is captured using a discrete HMM(with state emission), 
 \[
 M_0  = \langle S \cup S\times A, O, \trans, \emission, s_0 \rangle 
 \]
  \begin{itemize}
     \item $ S \cup S\times A $, including two types of states, a decision state $s$ at which an action will be selected, and a nature's state $(s, a)$ at which the next state will be determined according to a probability distribution.  
    \item  $  O$ is an alphabet, the set of observations;
    \item  $\trans: (S \cup S\times A)  \times (S \cup S\times A) \rightarrow[0,1]$ is the mapping defining the probability of each transition. The following constraints are satisfied: For $s\in S$, 
    \[
\trans (s, (s,a))= \pi_0(s, a );
    \]
    For $(s,a)\in S\times A$, 
    \[
\trans( (s,a),s' ) =  P( s,a, s');
    \]
   \item $\emission: (S \cup S\times A) \times O \rightarrow [0,1]$ is the mapping defining the emission probability of observation at a state that satisfies the following constraints:
    \[
    \emission(  s,o )  =\obs(o\mid s),
    \]
    and 
    \[
    \emission(  (s,a),o ) = \obs(o\mid s,a).
    \]
 \end{itemize}    
\end{definition}
It can be validated that the HMM is well-defined. 

The covert-planning problem is informally stated as follows.
\begin{definition}
    Given an \ac{mdp}  $M$ and a nominal behavior modeled as an HMM $M_0$, compute a policy $\pi$ for P1 that maximizes P1's total discounted reward, while ensuring, with a high probability $1-\alpha$ for $\alpha \in (0,1)$, P2 cannot detect P1's deviation from the nominal behavior $M_0$.
\end{definition}


 \section{Main Results}
\label{sec: main result}
First, we show that a finite-memory policy can be more powerful than a Markov policy for being $(1-\alpha)$-covert. 

\begin{theorem}
Given a \ac{mdp} $M$, P2's observation function $\obs$ and reward function $R$,  a \ac{hmm} $M_0$ modeling P2's observations given the nominal, Markovian policy $\pi_0$, and a parameter $\alpha \in (0,1)$, there may exist an  $(1-\alpha)$-covert finite-memory optimal policy that maximizes the total discounted reward but no $(1-\alpha)$-covert Markovian policy that can attain the same value of that finite-memory policy for the reward function.
    \end{theorem}

\begin{proof}

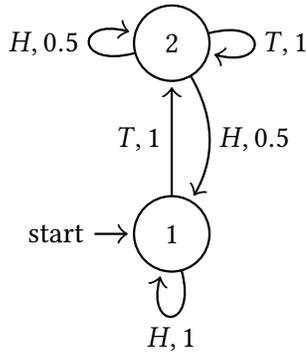
\begin{figure}[htbp]
	\centering
 \resizebox{0.5\linewidth}{!}{\begin{tikzpicture}[shorten >=1pt,node distance=2cm,on grid,auto, scale=0.6,transform shape]

  \node[state,initial] (1) {$1$};
  \node[state] (2) [above of=1] {$2$};
  \path [->](1) edge   [loop below]          node {$H, 1$} (1)
        (1) edge              node {$T, 1$} (2)
        (2) edge   [bend left]           node {$H, 0.5$} (1)
        (2) edge   [loop left]           node {$H, 0.5$} (2)
        (2) edge  [loop right]            node {$T, 1$} (2);
\end{tikzpicture}}

	\caption{The \ac{mdp} example}
	\label{fig:mdp}
\end{figure}
    We prove this theorem by constructing such a case. Consider a simple   \ac{mdp} with $S= \{1,2\}$, $A = \{H, T\}$ and stochastic transitions $P(1|1,H)=1$, $P(2|1, T)=1$, $P(1|2,H)=0.5$, $P(2|2, H) = 0.5$ and $P(2|2,T)=1$, illustrated in Fig.~\ref{fig:mdp}.

An agent is allowed to take $N$ actions and then stop. The nominal hypothesis is that at each state, the user selects ``$H$'' or ``$T$'' with equal probabilities.  The observer can observe the actions, but not the states. It is not hard to see that the hidden Markov model (for observations) induced by the nominal policy describes the sequence of outcomes for flipping a fair coin for $N$ times, in which the number of heads is a random variable that follows a binomial distribution.
The reward function for the planning agent is $R(1, H, 1) = 1$, $R(2, H, 1) = 1$, and for other (state, action, state) pairs, the reward is $0$. Therefore, if the agent is to maximize the total reward without covertness, it would select only action $H$ at both states.

Consider the observer employing hypothesis testing: With a confidence level $\alpha$ and total $N$ of actions, the observer would reject the nominal hypothesis if either the number of heads or the number of tails is greater than or equals an integer $K$. For example, if $\alpha=0.05$, $N=10$, then the null hypothesis is rejected if the number of heads is less than 2 or greater than 8.

Let $\kappa$ be the random variable that represents the number of times ``$H$'' is observed. Assuming $N = 2$ (the extension for $N\ge 2$ is discussed later), if the observer rejects the null hypothesis for observing more than one ``$H$'', then the covert constraint can be described as $Pr(\kappa \le 1) \ge 1-\rho$, where $\rho \in (0, 1)$ is a given threshold.

First, we consider the initial state to be $1$ and a Markovian policy $\pi_2$ defined by $\pi_2(H|1) = \alpha$, and $\pi_2(H|2) = \beta$ for $\alpha,\beta \in [0,1]$. Considering all trajectories generated by taking $2$ actions, which are $\{``111", ``112", ``121", ``122"\}$, the corresponding probabilities of generating the trajectories are $\alpha^2, \alpha(1-\alpha), (1-\alpha)\frac{\beta}{2}, (1-\alpha)(1-\beta) + (1-\alpha)\frac{\beta}{2}$ respectively. In order to find a Markovian policy that satisfies the covert constraint, we have $\alpha^2 \le \rho$. The corresponding value of the Markovian policy is $2\alpha^2 + \alpha(1-\alpha) + (1-\alpha)\frac{\beta}{2} = \alpha^2 + \alpha + (1-\alpha)\frac{\beta}{2}$. This value is smaller or equal to $\alpha^2 + \frac{\alpha + 1}{2}$, the equality holds when $\beta = 1$.

Next, we consider a finite-memory policy $\pi_2^\dagger$, where the agent takes action ``$H$'' at step $1$, and takes action ``$H$'' with probability $\rho$, action ``$T$'' with probability $1-\rho$ at step $2$. Note that this policy cannot be represented by a Markov policy. It is clear the finite-memory policy satisfies the covert constraint. And the corresponding reward is $1 + \rho$ in two steps.

Let $r_1$ denote the optimal Markovian policy's total reward, which is attained when $\beta=1$, and $r_2$ denote the finite-memory policy's total reward. Then  the following relation can be established: 
\begin{align}
   & r_2 - r_1 \\
   = & 1 + \rho -\alpha^2 - \frac{\alpha + 1}{2}\\
   \ge  & 1 - \frac{\alpha + 1}{2}\\
   \ge & \frac{1}{2} (1- \sqrt{\rho})
\end{align}
where the inequality is due to the covert constraint $\alpha^2 \le \rho$ and thus 
$- \alpha \ge -\sqrt{\rho}$. Since $\rho \in(0, 1)$, we have $r_2 - r_1 > 0$, which means the finite-memory covert policy performs better than the optimal Markovian covert policy when we take two actions. 

Next, we extend the case to $N$ actions for $N > 2$.  Using the hypothesis testing for binomial distribution, the observer would reject the nominal hypothesis if either the number of heads or that of tails exceeds an integer $K$. 
First, let $\pi_n$ denote the optimal Markovian policy that the agent can follow while satisfying the covert constraint. Let's construct a new policy $\pi'_n$ such that for  steps $1,2,\ldots, N-2$, $\pi'_n$ is exactly the same as $\pi_n$. At step $N-1$, if the agent is at state $2$, then the agent keeps following policy $\pi_n$; if the agent is at state $1$ and the number of ``H''  taken by the policy $\pi_n$ till this step is $X$, the decision for the next two actions will consider two different cases: 1) If $X+2 < K$, then the covert constraint allows two more ``H'' actions and the agent follows a finite-memory policy that takes two ``$H$'' actions. In this case,  any Markovian policy will have a reward at most the same as the finite-memory policy for the last two steps. 2) If $X+2\ge K$ and $X<K$, then the covert constraint allows only $0$ or $1$ ``$H$'' actions, then it is the case we discussed previously when $N = 2$, the best Markovian policy the agent can follow is $\pi_2$ at the last $2$ steps. However, the total reward obtained by following $\pi_2$   is smaller than the reward obtained by following  the finite-memory policy $\pi_2^\dagger$ given $N = 2$. The newly constructed $\pi_n'$ that commits to some finite memory policy in the last two steps is non-Markovian and satisfies the covert constraint. It also attains a better value than the optimal Markovian policy $\pi_n$ that satisfies the covert constraint. 
Based on this constructed example,  we conclude that a finite-memory policy can perform better than a Markovian policy when a covert constraint is enforced.
\end{proof}

 Despite the theorem showing that finite-memory policy is more powerful, it is intractable to compute because both the structure (memory states and transitions) and the mapping from the memory states to distributions over actions are unknown. Thus, in the next, we restrict our solution to Markovian policy space.

 \subsection{Computing a Covert Markovian Policy}

First, we review the sequential likelihood ratio test \cite{fuhSPRTCUSUMHidden2003}   for hypothesis testing of hidden Markov models. Based on the hypothesis testing methods, we then formalize the notion of $(1-\alpha)$-covert policy.

 The following notations are introduced: 
 We introduce a set of parameterized Markovian policies $\{\pi_\theta\mid \theta\in \Theta\}$ where $\Theta$ is a finite-dimensional parameter space. 
 For any Markovian policy  $\pi_\theta$ parameterized by  $\theta$, the Markov chain induced by $\pi_\theta $ from the \ac{mdp} $M$ is denoted $M_\theta \colon \{X_t, A_t, t\ge 0\}$ where $X_t$ is the random variable for the $t$-th state and $A_t$ is the random variable for the $t$-th action. For a   run $x =s_0a_0s_1a_1\ldots s_n$, $P(x; M_\theta)$ is the probability of the run $x$ in the Markov chain $M_\theta$.  We denote the distribution over a run by a random vector $X$ with support $\mathcal{X}$. Each sample of $X$ is a finite run $x$. 
\paragraph{Softmax Parameterization}
A natural class of policies is parameterized by the softmax function,
\begin{align}
    \pi_{\theta}(a|s) = \frac{\exp(\theta_{s,a})}{\sum_{a' \in A}\exp(\theta_{s, a'})}.
\end{align}
The softmax has good analytical properties including completeness and differentiability. It can represent any stochastic policy. In this work, we consider the policy to be in the softmax parameterization form.


 The observation distribution given policy $\pi_\theta$ is modeled by the \ac{hmm} that can be constructed similar to $M_0$ in Def.~\ref{def:hmm0}. We denote the observation as a random vector $Y$ with a support $\mathcal{Y}$. Each sample of $Y$ is a finite observation sequence $y = o_0o_1\ldots o_n$. The probability of observing $y$ given the policy $\pi_\theta$ is denoted $P(y;M_\theta)$. Likewise, $P(x;M_0)$(resp. $P(y; M_0)$) is the probability of a run $x$ (resp. an observation $y$) in the nominal model $M_0$,

Let's consider the case with an \emph{informed} defender who has access to the policy $\pi_\theta$ (the alternative hypothesis).  Under the Markov policy $\pi$, the stochastic process $\{O_i,i \ge 0\}$ is a hidden Markov model where $O_i$ is the random variable representing the observation at the $i$-th time step.  Let $Y_{1:n} \coloneqq O_0,O_1,O_2,\ldots, O_n$   be a  sequence of random variables for the  finite  observations with the unknown model $M$. 
The likelihood ratio is defined as:
\begin{align*}
S_n & \coloneqq \frac{P(Y_{1:n}; M_\theta)}{P(Y_{1:n};,\ldots, O_n; M_0)}.
\end{align*} 

The sequential probability ratio test (SPRT) of $M = M_0$ versus $M = M_\theta$ stops sampling at  the stage 
\[ 
T \coloneqq\inf\{n: \log S_n \le \epsilon \text{ or }  \log S_n \ge \beta\},
\]
where $\epsilon$ and $\beta$ are two parameters defined by SPRT and $\epsilon < \beta$. The test accepts the null hypothesis that $M= M_0$ if $\log S_T\le \epsilon$, and  accepts   the alternative hypothesis $M= M_\theta$ if $\log S_T\ge \beta$. Here, the thresholds $\epsilon, \beta$ are determined based on the bounds on  Type I and Type II errors in SPRT for hidden Markov models \cite{fuhSPRTCUSUMHidden2003}.

Thus,  given a single observation sequence $y$, the SPRT cannot reject the null hypothesis  if 
\[ \log \frac{P(y; M_\theta)}{P(y; M_0)} \le \epsilon. \]

Because each observation $y$ is generated with probability $P(y; M_0)$,  the following constraint enforces,  with a probability   greater than $\alpha$, that the null hypothesis not rejected: 
\[
\Pr(\log \frac{P(Y; M_\theta)}{P(Y; M_0)} > \epsilon; M_\theta) \le \alpha,
\]
where $\Pr(E ; M_\theta)$ is the probability of event $E$ in the hidden Markov model $M_\theta$.   For convenience, we refer the term $
\Pr(\log \frac{P(Y; M_\theta)}{P(Y; M_0)} > \epsilon; M_\theta)$ as the \emph{detection probability} and $\log \frac{P(Y; M_\theta)}{P(Y; M_0)} > \epsilon$ as the detection condition.

\begin{definition}
\label{def:covert-policy}
Given a small constant $\alpha \in [0,1]$, a   policy $\pi_\theta$ is $(1-\alpha)$-covert if it is the solution to the following problem.

\begin{alignat}{2}
    & \optmax_{\pi_\theta} &&\quad V^{\pi_\theta} (s_0 ) \label{eq:alpha-covert-obj}
\\
 & \optsts &&  \Pr \left( \log\frac{P(Y;M_\theta)}{P(Y;  M_0)}   
>  \epsilon    ; M_\theta \right) \le \alpha, 
\end{alignat}
where $M_\theta$ is the \ac{hmm} induced by policy $\pi_\theta $ given the \ac{mdp} $M$.
\end{definition} 
\begin{remark}
In Def.~\ref{def:covert-policy}, we conservatively assume an informed observer who has access to the agent's policy. In reality, the observer is not informed and his detection performance can be worse than the informed observer. The assumption of an informed observer is also common in minimal information-leakage communication channel design to ensure strong information security and privacy \cite{khouzaniLeakageMinimalDesignUniversality2017}. In our case, this assumption provides a strong guarantee for covertness.
\end{remark}

\subsection{Primal-dual Proximal Policy Gradient for  Covert Optimal Planning}
By the method of Lagrange multipliers \cite{bertsekas2014constrained}, we can formulate the problem in \eqref{eq:alpha-covert-obj} into an unconstrained max-min problem. For our problem, the Lagrangian function is given by
$$
 L(\theta,\lambda) = V(s_0, \theta) + \lambda \left(\alpha - \Pr \left( \log\frac{P(Y;M_\theta)}{P(Y;  M_0)}   
 >  \epsilon   ; M_\theta \right)\right),
$$
where $\theta$ is the primal variable and $\lambda \ge 0$ is the dual variable. Here, we replace the policy $\pi_\theta$ with the policy parameter $\theta$ for clarity and also rewrite $V^{\pi_\theta}(s_0)$ as $V(s_0, \theta)$.    The original constrained optimization problem in~\eqref{eq:alpha-covert-obj} can be reformulated as
$$
\optmax_{\theta}\quad \min_{\lambda \ge 0 } L(\theta, \lambda).
$$
%
%
%

Because  the value function can be a non-convex function of the policy parameters, we present a primal-dual gradient-based policy search method aiming to compute a (locally) optimal policy that satisfies the constraint. This method uses two-time-scale updates for the primal and dual variables. The primal variable $\theta$ is updated at a faster time scale, while the dual variable $\lambda$ is updated at a slower time scale (one update after several primal updates).

When performing the gradient computation of $L$ with respect to $\theta$,   it is observed that the constraint involves taking expectation over a distribution that depends on the decision variable $\theta$. To mitigate the distribution shift \cite{drusvyatskiy2023stochastic} in stochastic optimization, we then introduce a proximal policy gradient to bound the distribution shift between two updates on the primal variable.
That is, 
   at the $t$-th iteration, when computing the gradient of $ L(\theta, \lambda)$ with respect to $\theta$, we include a KL-divergence between the trajectory distribution $P_{\theta_t}$ under the current policy parameterized by $\theta_t$ and the trajectory distribution $P_{\theta }$ under the new policy parameterized by $\theta$, similar to the proximal policy optimization method \cite{schulman2017proximal}. After including the KL-divergence, the Lagrangian function becomes 
\begin{multline*}
 L^\beta (\theta,\lambda) = V(s_0, \theta) + \lambda \left(\alpha - \Pr \left( \log\frac{P(Y;M_\theta)}{P(Y;  M_0)}    
 - \epsilon > 0 ; M_\theta \right) \right) \\ - \beta D_{KL}(P_{\theta_t} ||P_{\theta}),
\end{multline*}
where $\beta >0$ is  a scaling parameter that can be dynamically adjusted and \[ D_{KL}(P_{\theta_t} ||P_{\theta}) = \Expect_{x \sim P_{\theta_t}}\left(\log\frac{ P_{\theta}(x)}{P_{\theta_t}(x)} 
\right) = \sum_{x\in \mathcal{X}} P_{\theta_t}(x)\log(\frac{ P_{\theta}(x)}{P_{\theta_t}(x)}),\] 
where $P_\theta(x)$ (resp. $P_{\theta_t}$) is the probability of the run $x$ in the Markov chain $M_\theta$ (resp. $M_{\theta_t}$).

Taking the gradient of $L^\beta(\theta, \lambda)$ with respect to $\theta$,  note that $\Pr \left( \log\frac{P(Y;M_\theta)}{P(Y;  M_0)}   
>  \epsilon    ; M_\theta \right) = \Expect_{y\sim M_\theta} \left[ 
 \mathbf{1}( \log\frac{P( y; M_\theta)}{P(  y; M_0)}   
 - \epsilon > 0 ) \right], $ where $\mathbf{1}(E)$ is the indicator function that evaluates to one if $E$ is true and zero otherwise. We have
\begin{multline*}
\nabla_\theta L^\beta (\theta,\lambda) = \lambda \nabla_\theta \left(\alpha - \Expect_{y\sim M_\theta} \left[ 
\mathbf{1}( \log\frac{P(Y= y; M_\theta)}{P(Y = y;  M_0)}   
 - \epsilon > 0 ) \right] \right) \\ + \nabla_\theta V(s_0, \theta) 
- \beta \nabla_\theta D_{KL}(P_{\theta_t} ||P_\theta).
\end{multline*}

We approximate the gradient using samples generated from the current chain $M_{\theta_t}$.
Each sample is a pair $(x_i, y_i)$ that includes: 1) a run  $x_i=s_{i,0}a_{i,0}s_{i,1}a_{i,1}\ldots s_{i,T},$
where $T $ is the length of state sequence, and 2) an observation $y_i$,   sampled from the probability distribution $P(Y\mid x_i, M_{\theta_t}) $.  For each $x_i \in \mathcal{X}_N$, let $R(x_i) =  \sum_{t=1}^T {\gamma}^t R(s_{i,t}, a_{i, t})$ be the total discounted rewards accumulated with the run $x_i$.

 First, we compute the gradient of the value function with respect to the policy parameter $\theta$, 
\begin{align}
    \nabla_\theta V(s_0, \theta) & = 
 \nabla_\theta    \sum_{x\in \mathcal{X}}  P_{\theta }(x)  R(x) \\
    & =  \sum_{x\in \mathcal{X}} \nabla_\theta \left(P_{\theta_t}(x)\frac{P_{\theta}(x)}{P_{\theta_t}(x)}\right)R(x) \\
    & = \sum_{x\in \mathcal{X}} P_{\theta_t}(x) \frac{\nabla_{\theta}P_{\theta}(x)}{P_{\theta_t}(x)}R(x) \\
    & = \sum_{x\in \mathcal{X}} P_{\theta_t}(x) \frac{P_{\theta}(x)}{P_{\theta_t}(x)}\nabla_{\theta}\log (P_{\theta}(x))R(x) \\
    & \approx \frac{1}{N}\sum_{x_i \in \mathcal{X}_N}\frac{P_{\theta}(x_i)}{P_{\theta_t}(x_i)}\left(\sum_{t=1}^T\nabla_\theta \log\pi(a_{i,t}|s_{i,t})\right)R(x_i),
\end{align}
where $\mathcal{X}_N =\{x_i, i=1,\ldots, N\} \subseteq \mathcal{X}$ be a sample of $N$ finite-length runs.   Because the samples are obtained with policy $\pi_{\theta_t}$, we use importance weighting in the second step. In the last step, we approximate the gradient using the $N$ sampled trajectories.

The third term is the KL-divergence. Taking derivate with respect to $\theta$, we have:
\begin{align}
\nabla_\theta D_{KL}(P_{\theta_t} ||P_{\theta}) & =\sum_{x\in \mathcal{X}}\left( \nabla_\theta P_{\theta_t}(x)\log\frac{ P_{\theta_t}(x)}{P_{\theta}(x)}\right) \\ 
& = -\sum_{x \in \mathcal{X}} P_{\theta_t}(x) \nabla_\theta \log P_{\theta}(x)\\
& = -\Expect_{x \sim P_{\theta_t}}\left[\nabla_\theta\log(P_{\theta}(x)) \right] \\
& \approx -\frac{1}{N}\sum_{x_i\in \mathcal{X}_N}\left(\sum_{t=1}^T\nabla_\theta \log\pi_\theta (a_{i,t}|s_{i,t})\right).
\end{align}

 For the first term, we need to compute the derivative of the constraint with respect to the policy parameter. 
\begin{align}
     & \nabla_\theta \left(\alpha - \Expect_{y\sim M_\theta} \left[ 
\mathbf{1}( \log\frac{P( y; M_\theta)}{P(  y; M_0)}   
 - \epsilon > 0 ) \right] \right) \\
 = & -\nabla_{\theta} \sum_y P(  y; M_\theta)\left[ \mathbf{1}(\log\frac{P(  y; M_\theta)}{P(  y; M_0)} - \epsilon > 0)\right] \\
 = &-\sum_y \nabla_{\theta} P(  y; M_\theta)\left[ \mathbf{1}(\log\frac{P( y; M_\theta)}{P(  y; M_0)} - \epsilon > 0)\right] \\
 = & -\sum_{y \in U} \nabla_{\theta} P(  y; M_\theta) \\
 = & -\sum_{y \in U} \nabla_{\theta} \sum_{x \in \mathcal{X}} P(y | x)P_{\theta}(x) \\
 = & -\sum_{y \in U}\sum_{x\in \mathcal{X}} P(y|x) \nabla_{\theta} P_{\theta}(x), \end{align}
 where $U = \{y\in \mathcal{Y} \mid  \log\frac{P( y; M_\theta)}{P( y; M_0)} - \epsilon > 0\}$ is a set of observation sequences at which the detection condition is met.

 Using the logarithm trick again, we have
\begin{align} 
 & \sum_{y \in U}\sum_{x\in \mathcal{X}} P(y|x) \nabla_{\theta} P_{\theta}(x)\\
= & \sum_{y \in U} \sum_{x\in \mathcal{X}} P(y|x) P_{\theta}(x)\nabla_{\theta}\log(P_{\theta}(x)) 
\label{eq:gradient-detection-step1}  \\
 = & \sum_{y \in U} \sum_{x\in \mathcal{X}} P_{\theta_t}(x) \frac{P_{\theta}(x)}{P_{\theta_t}(x)} P(y|x) \nabla_{\theta} \log P_{\theta}(x) \label{eq:gradient-detection-step2} \\
\approx &  \frac{1}{N} \sum_{x_i \in \mathcal{X}_N} \sum_{y \in U} \frac{P_{\theta}(x_i)}{P_{\theta_t}(x_i)}P(y|x_i) \nabla_{\theta}\log P_\theta(x_i). \label{eq:approx-1}
\end{align} 
The approximation requires the computation of $P(y|x_i)$, which is the probability of observation $y$ given the run $x_i$, for all $y \in U$. In practice, the set $U$ can be large and difficult to construct. We discuss two approaches to compute the approximation. 

One approach is to uniformly sample a subset of $U$, called $U_K$, and compute the gradient as 
\begin{align*}
\text{\eqref{eq:gradient-detection-step1}}   \approx  &  \frac{1}{N} \sum_{x_i \in \mathcal{X}_N} \sum_{y_k \in U_K} \frac{P_{\theta}(x_i)}{P_{\theta_t}(x_i)}P(y_k|x_i) \nabla_{\theta}\log P_\theta(x_i)  \\
= 
    & \frac{1}{N} \sum_{x_i \in \mathcal{X}_N} \sum_{y_k \in U_K} \frac{P_{\theta}(x_i)}{P_{\theta_t}(x_i)}P(y_k|x_i)  \nabla_{\theta} (\sum_{t=1}^T\nabla_\theta \log\pi_\theta (a_{i,t}|s_{i,t})).\end{align*}

Alternatively, it is noted that $P(y|x)P_{\theta_t}(x)$ is the joint probability $P_{\theta_t}(x,y)$.
\begin{align*}
  \text{\eqref{eq:gradient-detection-step2}} =  & \sum_{x\in \mathcal{X}}\sum_{y \in U} P_{\theta_t}(y,x) \cdot   \frac{P_\theta(x)}{P_{\theta_t}(x)} \cdot \nabla_{\theta} \log P_\theta(x) \\
\approx &\frac{1}{N} \sum_{x_i\in \mathcal{X}_N}\mathbf{1}(y_i \in U)   \frac{P_\theta(x_i)}{P_{\theta_t}(x_i)} \cdot \nabla_{\theta} \log P_\theta(x_i) \\
=  &\frac{1}{N} \sum_{x_i\in \mathcal{X}_N}\mathbf{1}(y_i \in U)   \frac{P_\theta(x_i)}{P_{\theta_t}(x_i)} \cdot \left(\sum_{t=1}^T\nabla_\theta \log\pi_\theta (a_{i,t}|s_{i,t})\right).
 \end{align*}
 That is, for each sampled trajectory-observation pair $(x_i, y_i)$ from the distribution $P_{\theta_t}$, we check if $\log \frac{P_\theta( y)}{P_0(y)} -\epsilon > 0$,  that is, using the distribution $P_\theta$ to determine if $y_i\in U $.

Finally, the gradient of the dual variable $\lambda$ is calculated as the following: 

\begin{align}
& \nabla_\lambda L^\beta(\theta, \lambda)\\
= & \alpha - \Pr \left( \log\frac{P(Y;M_\theta)}{P(Y;  M_0)}   
 - \epsilon > 0 , M_\theta \right) \\
 = &  \alpha - \mathbf{1}\left(\log\frac{P(Y;M_\theta)}{P(Y;  M_0)} >\epsilon\right) \\
 \approx & \alpha - \frac{1}{N} \sum_{y \in \bar U}\mathbf{1}\left(\log\frac{P( y; M^{\theta})}{P( y;  M_0)} >\epsilon\right) ,
\end{align}
where $\overline{U} = \{y_i \mid\log\frac{P( y_i; M_{\theta})}{P(  y_i;  M_0)} >  
 \epsilon\}$ is a set of \emph{sampled} observations at which the detection condition is met.

 Algorithm~\ref{alg:PI} summarizes the proposed Primal-Dual Covert Policy Gradient (Covert PG) method using two-time-scale updates.   Two-time-scale updates are a common technique used by first-order algorithms for solving minimax or maximin problems (see, e.g.,~\cite{nouiehed_solving_2019,thekumparampil_efficient_2019,lin_near-optimal_2020,jin_what_2020,fiez2021local}). The inner-loop updates are carried out more frequently in order to compute an approximately optimal solution to the inner-loop optimization problem, which can be subsequently used for computing the gradient for the outer-loop updates. The convergence analysis of the proposed primal-dual proximal policy gradient method is left as future work. While several convergence analyses of two-time-scale first-order methods exist, they do not directly apply to the problem studied in this paper due to different assumptions on the optimization problem such as strong convexity/concavity of the inner problem~\cite{lin_near-optimal_2020} or that the primal-dual solution forms a strict local Nash equilibrium~\cite{jin_what_2020,fiez2021local}. 
 
 In our numerical experiments, we chose to terminate the algorithm once the value of the Lagrangian did not change significantly between two consecutive iterations, which is determined by the choice of $\delta_0$. The parameter $d$ is the KL divergence target distance. If the KL divergence is significantly different from the target distance, $\beta$ quickly adjusts. Finally, $(x)^+ = 0$ if $x \le 0$, $(x)^+ = x$ if $x > 0$.

\begin{algorithm}[hbt!]
\caption{Primal-dual proximal policy gradient for covert optimal planning}\label{alg:PI}
\begin{algorithmic}[1]
\Procedure {Primal-dual proximal policy gradient}{$\lambda_1, \theta_1, \beta_1$}
\State $t \gets 1$
\State $L^\beta(\theta_0, \lambda_0) \gets \infty$
\State $\delta \gets \infty$
\While {$| \delta| \ge \delta_0$}
\For {\texttt{batches $b = 1, 2,...m$}}
    \State Generate trajectories $\mathcal{X}^t_b$ based on $\theta_t$
    \State Calculate gradient term $\nabla_\theta L^\beta (\theta,\lambda_t)$ given generated trajectories $\mathcal{X}_b$
    \State $\theta \gets \theta + \eta \nabla_\theta L^\beta (\theta,\lambda_t)$
\EndFor
\State Calculate $L^{\beta}(\theta, \lambda_t)$
\State Calculate $\nabla_\lambda L^\beta(\theta, \lambda_t)$ based on all trajectories $\cup_b\mathcal{X}^t_b$
\State  $\lambda_{t+1} \gets (\lambda_{t}- \kappa \nabla_\lambda L^\beta(\theta  , \lambda_t ))^+$ 
\State Calculate $D_{KL}(P_{\theta_t} ||P_{\theta})$
\If{$D_{KL}(P_{\theta_t} ||P_{\theta}) \le d/1.5$}
    \State $\beta \gets \beta/2$
\EndIf
\If{$D_{KL}(P_{\theta_t} ||P_{\theta}) \ge d \times 1.5$}
\State $\beta \gets \beta \times 2$
\EndIf
\State $\theta_{t+1} \gets \theta$
\State $\delta \gets | L^{\beta}(\theta, \lambda_t) - L^{\beta}(\theta, \lambda_{t-1}) |$
\State $t \gets t+1$
\EndWhile

\EndProcedure
\end{algorithmic}
\end{algorithm}






\section{Experiment validation}
\label{sec: experiment}

We demonstrate our solutions using a  gridworld planning environment,  depicted in Figure~\ref{fig:10X10 gridworld}. A state of the agent is denoted by $(\mbox{row}, \mbox{col})$. The adversary can move in one of four compass directions. When given the action ``N'', the adversary enters the intended cell with a $1 - 2\beta$ probability and enters the neighboring cells, which are the west and east cells, with probability $\beta$. Different values of stochasticity parameter $\beta$ are used in the experiment.
The purple areas are bouncing walls, which the agents can not enter. If the agent takes an action that crashes into the bouncing walls or boundaries, the agent will stay in his previous state. The black cells are penalty states, in which the agent receives a penalty if he enters the state.  The robot icon in the figure indicates the agent's initial state (for both the nominal user and the adversary). 

For both the adversary and the normal user, the cost of each action is $0.2$. They incur a cost of $2$ if they enter the black states and receive a reward of $20$ upon reaching their respective goals. The discounting factor is $0.95$. The normal user's goal is  $(7, 0)$ (marked by the star icon) and the adversary's goal is located at $(0, 0)$ (marked by the red circle). 
The normal user's policy is computed as the entropy-regulated optimal policy to maximize the total discounted reward with respect to the normal user's reward function. 

We consider the behavior of the agent is observed using a set of sensors, 
whose configuration and the covering areas are shown in Figure~\ref{fig: sensor config}. The sensor icon indicates the sensor's location and the yellow area depicts the sensor coverage. Each sensor possesses a default detection probability and can solely detect agents presented within its coverage. The detection probability decreases with the distance of the agent's location to the sensor's location.   Once a sensor detects an agent, the defender obtains an observation value of $1$ from that particular sensor. Consequently, our experimental setup results in the defender's observation taking the form $(o_1, o_2, o_3)$, with each $o_i$ corresponding to the observation from sensor $i$.

The detection probability for a given state is influenced by three key factors: the range of the sensor, the sensor's default detection probability ($0.8$ for all sensors), the distance between the current state and the sensor's location, and the surrounding environment (green states). 
For every unit of distance increase, the detection probability decreases by $0.05$. 
A dark (resp. light) green state further decreases the probability by $d_1= 0.2$ (resp. $d_2= 0.1$). For example, assume sensors one, two, and three are located at $(4, 0), (6, 4), (1, 4)$ respectively, if an agent is at state $(6,3)$, which is not in the range of sensor one and three, thus the observations received from sensor one and three are $0$. For sensor 2,   the probability of detecting the agent is $0.8-0.05 \times 1 - 0.2=0.55$. 

With this stochastic environment and sensor setup, we compute the optimal $(1-\alpha)$-covert policy given $\alpha=0.2$ and different system stochasticity parameters $\beta = 0.05, 0.1, 0.15$.  In the experiments, the learning rate of $\theta$ is $0.005$, and the learning rate of $\lambda$ is $0.01$. In each iteration, $800$ trajectories are generated. These trajectories are divided into $20$ batches to update $\theta$  in the primal-dual proximal policy gradient computation (Line 5-9 in Algorithm~\ref{alg:PI}). 

\begin{figure}[ht]
	\centering
	\includegraphics[width = 0.5\linewidth]{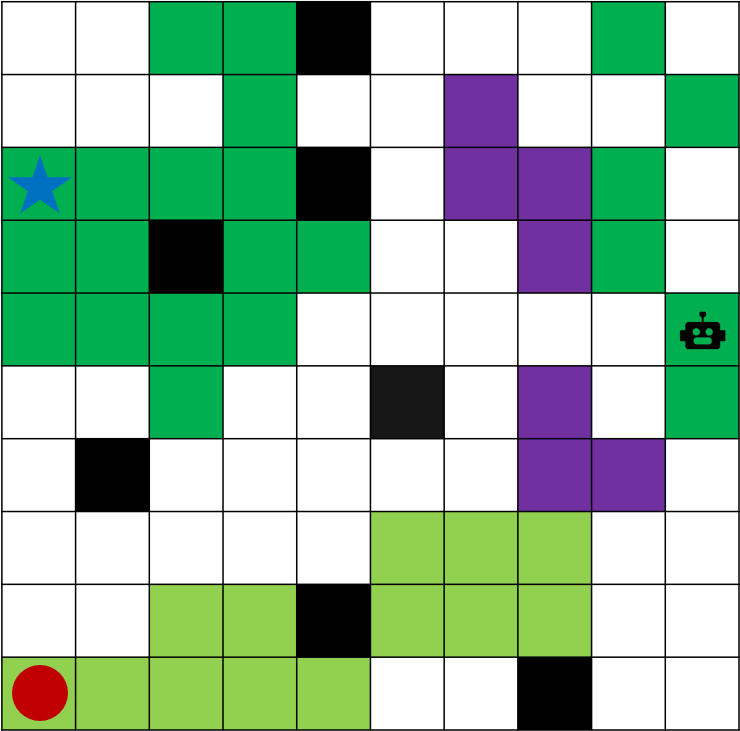}
	\caption{The 10 $\times$ 10 stochastic gridworld.}
	\label{fig:10X10 gridworld}
\end{figure}

\begin{figure}[ht]
	\centering
	\includegraphics[width = 0.5\linewidth]{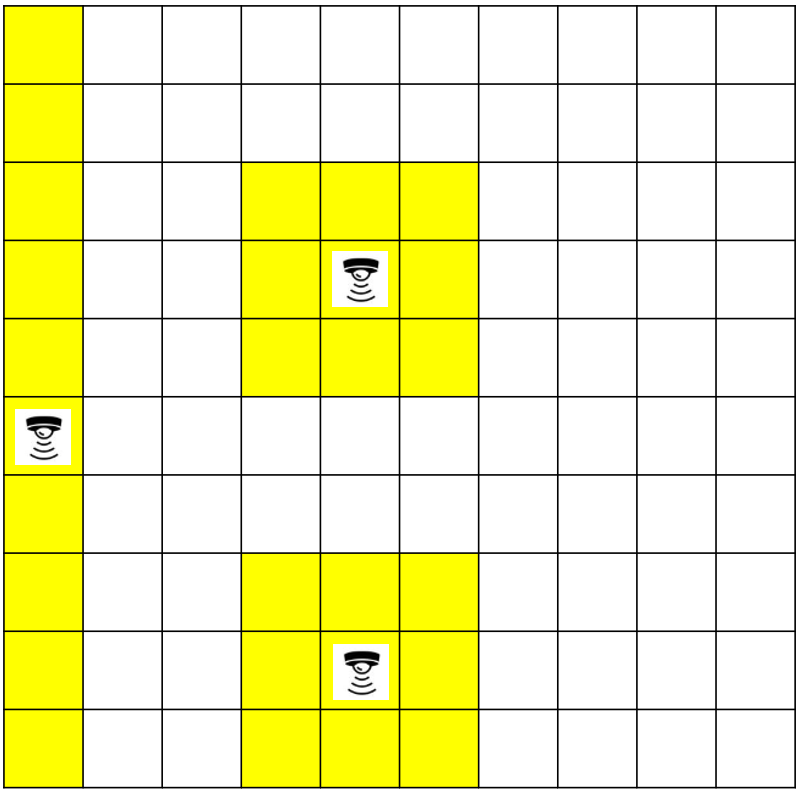}
	\caption{The Sensor Configuration.}
	\label{fig: sensor config}
\end{figure}

First, we compute the optimal policies without the covert constraint and evaluate the probability of detection of these policies. 
Given a system's stochasticity parameter $\beta = 0.05$, without the covert constraint, the adversary's deterministic optimal policy would attain a value of $6.76$ at a cost of being detected with a probability close to  $99\%$. Similarly, the adversary's value is $5.6$ and $4.2$ given $\beta = 0.1$,  and  $\beta = 0.15$ respectively if the adversary takes  deterministic optimal policies in these two environments. The detection probabilities are close to $99\%$ under both cases. If the adversary uses the softmax optimal policy, then the values are close to the deterministic optimal policies but the detection probabilities are approximately $73\%$, given $\beta = 0.05, 0.1, \text{and } 0.15$.

Then, we solve for  $0.8$-covert policy that ensures the probability of detection is no greater than $\alpha=20\%$.
We set the detection threshold $\epsilon$ to be $3$, following from $\chi^2$ distribution with 1 degree-of-freedom and confidence interval $92\%$.
It is conservatively selected to model a  detector that tolerates a relatively large false positive rate.  The performance of the algorithm is empirically analyzed from three values: Lagrangian function value, adversary's expected value, and detection probability, as shown in Figures ~\ref{fig: Lagrangian-converge-trend}, ~\ref{fig: expected-value-converge trend}, and ~\ref{fig: detection-convergence-trend}. We select the entropy-regulated optimal policy as the initial policy. Thus, the values of policies are highest upon initialization. As the policy is updated over iterations to enforce the covertness constraint, the values decrease. The algorithm terminates in all three cases. 

 \begin{figure}[htbp]
	\centering
	\includegraphics[width = 1\linewidth]{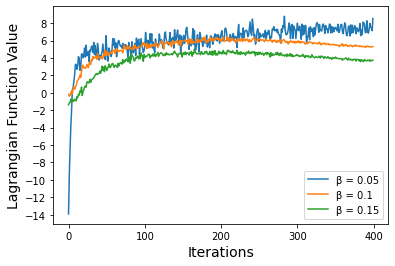}
	\caption{Lagrangian value over iterations given $\beta = 0.05, 0.1, 0.15$.}
	\label{fig: Lagrangian-converge-trend}
\end{figure}

 Figure~\ref{fig: Lagrangian-converge-trend} shows the Lagrangian value change over iterations given $\beta = 0.05, 0.1 \text{and } 0.15$. When $\beta = 0.1, \text{and } 0.15$, the initial value of $\lambda$ is $10$. The  values of the Lagrangian functions change initially increases and then decreases after 200 iterations. This is because the detection constraint is satisfied after 200 iterations and after that the value of $\lambda$ decreases to $0$, resulting in a decrease in the Lagrangian value.  
 It is observed that the initial value of Lagrangian for $\beta=0.05$ is smallest, close to $-14$. This is because we have picked the initial value of $\lambda$ to be $40$ for $\beta=0.05$ to achieve a faster termination.   
 

 \begin{figure}[htbp]
	\centering
	\includegraphics[width = 1\linewidth]{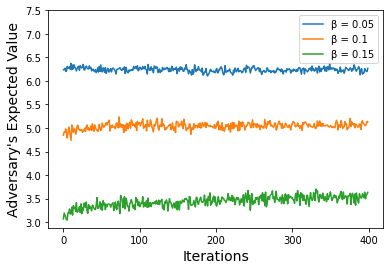}
	\caption{Expected value over iterations given $\beta = 0.05, 0.1, 0.15$.}
	\label{fig: expected-value-converge trend}
\end{figure}

From Figure~\ref{fig: expected-value-converge trend}, we conclude the adversary's expected values are $6.3, 5.1, \text{and } 3.47$, given $\beta = 0.05$, $0.1$, and $0.15$ respectively. Compare the adversary's covert policies' expected values to the adversary's optimal expected values without enforcing covertness constraint ($6.76$, $5.6$, and $4.2$), the decreases are $7.5\%, 9\%, \text{and } 18\%$, for $\beta = 0.05, 0.1, \text{and } 0.15$ respectively.  This shows that in this current environment setting, 
the adversary can satisfy the detection constraint without too much loss in the task performance. 

\begin{figure}[htbp]
	\centering
	\includegraphics[width = 1\linewidth]{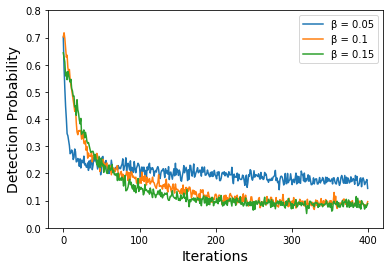}
	\caption{Detection probability over iterations given $\beta = 0.05, 0.1, 0.15$.}
	\label{fig: detection-convergence-trend}
\end{figure}

Figure~\ref{fig: detection-convergence-trend} depicts the change of the detection probabilities over iterations. The result shows all three initial policies have large detection probabilities, and the detection probabilities decrease quickly in the first $100$ iterations and eventually satisfy the detection constraint. The final detection probabilities are $0.168, 0.089, \text{and } 0.088$, given $\beta = 0.05, 0.1, \text{and } 0.15$ respectively. The policies obtained upon termination all satisfy the $0.8$-covertness.

To test the sensitivity of the policies concerning different system stochasticities, we evaluated the detection probabilities of the computed optimal policies for $\beta=0.05,0.1$ and $0.15$, under different levels of system stochasticity. The results are presented in Table~\ref{tab: detection given stochasticity}. We observe that if the policy is optimized for the system with a  stochasticity level $\beta$, then this policy can still ensure covertness for $\beta' \ge \beta$ (as shown in Boldface in the table). We hypothesize that the increased noise can aid the covertness of the policy. However, verifying this hypothesis requires further analysis and more general classes of \ac{mdp}s other than the gridworld dynamics. We leave this to future investigation.


The experiments are conducted using Python on a Windows $10$ machine with Intel(R) Core (TM) i7-11700K CPU and 32~GB RAM. The average running time for $400$ iterations is approximately $24$ hours. The running time varies slightly with different system stochasticities due to differences in trajectory length.

\begin{table}[htbp]
\resizebox{\columnwidth}{!}{%
\begin{tabular}{|l|l|l|l|}
\hline
$\beta$                    & 0.05  & 0.1   & 0.15  \\ \hline
$\pi^{\ast}$ given $\beta = 0.05$ & $\mathbf{0.168 \pm 0.011}$ & $\mathbf{0.077 \pm 0.006}$ & $\mathbf{0.083 \pm 0.007}$ \\ \hline
$\pi^{\ast}$ given $\beta = 0.1$  & $0.289 \pm 0.018$ & $\mathbf{0.089 \pm 0.012}$ & $\mathbf{0.087 \pm 0.009}$ \\ \hline
$\pi^{\ast}$ given $\beta = 0.15$ & $0.332 \pm 0.011$ & $0.093 \pm 0.009$ & $\mathbf{0.088 \pm 0.008}$ \\ \hline
\end{tabular}
}
\caption{The detection probabilities  obtained by evaluating different policies in different stochastic dynamics for the system.}
\label{tab: detection given stochasticity}
\end{table}

\section{Conclusion}
\label{sec: conclusion}
We study a class of covert planning problems against imperfect observers and develop a covert primal-dual gradient method that optimizes task performance given certain covertness constraints. Despite the proof that finite-memory policies can be more powerful than memoryless policies under the covertness constraint, our solution is limited to finding an optimal and covert Markov policy. Future work may investigate approaches to efficiently search finite-memory policy space under covert constraints.  By analyzing the covert policy, observer can consider whether it is possible to use the covert policy as a counterexample to improve the imperfect observer and eliminate such covert policies.  Additionally, through empirical analysis, one can investigate theoretically how the detection probability can be influenced by the stochasticity in the system and different levels of noise in the observations.

\section*{Acknowledgement}

Research was sponsored by the Army Research Office  under Grant
Number W911NF-22-1-0034 and the Army Research Laboratory   under Cooperative Agreement Number W911NF-22-2-0233. The views and conclusions
contained in this document are those of the authors and should not be interpreted as representing the official
policies, either expressed or implied, of the Army Research Laboratory or the U.S. Government. The U.S.
Government is authorized to reproduce and distribute reprints for Government purposes notwithstanding
any copyright notation herein.


\newpage
 \bibliographystyle{ACM-Reference-Format} 
\bibliography{sample}


\end{document}